\documentclass[11pt]{article}
\usepackage{macros} 

\usepackage{setspace}
\usepackage[export]{adjustbox}
\usepackage{makecell}
\usepackage{tqft}
\usepackage{enumitem}
\usetikzlibrary{graphs}
\usepackage{forest}

\usepackage{authblk}

\mathtoolsset{showonlyrefs}

\title{\textbf{Improved rate-distance trade-offs for quantum codes with restricted connectivity}}
\author[1]{Nou\'edyn Baspin}
\affil[1]{Centre for Engineered Quantum Systems, School of Physics, University of Sydney, New South Wales 2006, Australia}
\author[2]{Venkatesan Guruswami}
\affil[2]{Departments of EECS and Mathematics, and the Simons Institute for the Theory of Computing, UC Berkeley, Berkeley, CA, 94709, USA}
\author[3]{Anirudh Krishna}
\affil[3]{Department of Computer Science and Stanford Institute for Theoretical Physics, Stanford University, Stanford, CA, 94305, USA}
\author[4]{Ray Li}
\affil[4]{Department of EECS, UC Berkeley, Berkeley, CA, 94709, USA}

\date{June 2023}

\begin{document}

\maketitle

\begin{abstract}
    For quantum error-correcting codes to be realizable, it is important that the qubits subject to the code constraints exhibit some form of limited connectivity. 
    The works of Bravyi \& Terhal~\cite{bravyi2009no} (BT) and Bravyi, Poulin \& Terhal~\cite{bravyi2010tradeoffs} (BPT) established that geometric locality constrains code properties---for instance  $\dsl n,k,d \dsr$ quantum codes defined by local checks on the $D$-dimensional lattice must obey $k d^{2/(D-1)} \le O(n)$.
    Baspin and Krishna~\cite{baspin2022connectivity} studied the more general question of how the connectivity graph associated with a quantum code constrains the code parameters.
    These trade-offs apply to a richer class of codes compared to the BPT and BT bounds, which only capture geometrically-local codes.
    We extend and improve this work, establishing a tighter dimension-distance trade-off as a function of the size of separators in the connectivity graph.
    We also obtain a distance bound that covers all stabilizer codes with a particular separation profile, rather than only LDPC codes.
\end{abstract}

\section{Introduction}
We refer to an $\dsl n,k,d \dsr$ quantum code $\cQ$ when the code $\cQ$ uses $n$ physical qubits to encode $k$ qubits.
In other words, $\cQ$ is a $2^k$-dimensional subspace of $\bbC^{2^n}$.
Furthermore, the code is robust to any set of erasures of size strictly less than the distance $d$.
One way to represent the error correcting code $\cQ$ is using a set of generators $\cS = \{\ssS_1,...,\ssS_m\}$, a commuting set of Pauli operators.
The code space is the simultaneous $+1$-eigenspace of these generators.
The generators play the role of parity checks for classical codes; measuring the generators (via a syndrome-extraction circuit) yields a binary string $\bs = (s_1,...,s_m)$ which can be used to deduce and correct errors.

A coarse way of capturing the relationship between the generators $\cS$ and the qubits of $\cQ$ is via the \emph{connectivity graph} $G$.
The vertices $V$ of the connectivity graph $G$ correspond to the $n$ qubits of $\cQ$ and two qubits $q,q'$ are connected by an edge if they are both in the support of some generator $\ssS \in \cS$.
This relationship is coarse in that we discard information about what Pauli operator acts on a qubit---if qubits $q$ and $q'$ are in the support of $\ssS$, we draw an edge between $q$ and $q'$ regardless of whether the restriction of $\ssS$ to $q$ is $\ssX$ or $\ssZ$.
Nevertheless, there is an intimate relationship between the code $\cQ$ and some graph invariants of $G$.
This is the central motivation in our paper.

In addition to being a useful theoretical tool, the study of graph invariants of $G$ can impact the design and implementation of the syndrome-extraction circuit for $\cQ$ in two dimensions.
In this case, the edges of the graph correspond to two-qubit gates.
Physical architectures might be limited in many ways---for example, in the length of reliable two-qubit gates that can be implemented, the number of two-qubit gates a qubit can be involved in or the number of two-qubit gates that can be implemented in parallel.

The exploration of the relationship between codes and the associated connectivity graphs begins with the seminal work of Bravyi \& Terhal \cite{bravyi2009no}.
They studied \emph{geometrically-local} codes, i.e.\ codes for which qubits are embedded in $D$-dimensional Euclidean space and all generators $\ssS_i \in \cS$ are supported within a ball of radius $w = O(1)$.
Bravyi \& Terhal proved that the distance of any $D$-dimensional local code is limited---they showed that
\begin{align*}
    d = O(n^{1-1/D})~.
\end{align*}
In a beautiful paper shortly thereafter, Bravyi, Poulin \& Terhal \cite{bravyi2010tradeoffs} showed geometric locality also imposes a sharp trade-off between $k$ and $d$.
When $D=2$, they showed that
\begin{align}
\label{eq:2d-bpt}
    kd^2 = O(n)~.
\end{align}
We henceforth refer to this result as the BPT bound.

While this bound proves strict geometric-locality constrains code parameters, what is the best we can do with some limited set of long-range two-qubit gates?
Conversely, what sorts of graphs does one have to realize in $2$ dimensions to implement a code with some target parameters?
These questions were partially addressed in \cite{baspin2022connectivity}.
The \emph{separation profile} $s_G$ of the graph $G$ was identified as an important metric.

To define this quantity, we begin with the graph separator.
The separator $\mathsf{sep}(G)$ is a set of vertices in $V$ such that we can write $V = A \sqcup \mathsf{sep}(G) \sqcup B$.
Furthermore, there are no edges between $A$ and $B$ and both partitions obey $|A|, |B| \leq 2n/3$.
The separator $\mathsf{sep}(G)$ itself is not unique; we define $s_G$ via a minimization over all separators:
\begin{align}
    s_G = \min_{\sep(G)}\left(|\sep(G)|\right)~.
\end{align}
Studying the separator helps us understand the connectivity of the connectivity graph---it formalizes the intuition that the more difficult it is to partition the graph into two pieces, the more connected it is.
To motivate the separation profile, we note that some portions of $G$ can be well connected while others can be poorly connected.
The separator itself is a bit crude as it only quantifies the connectivity of the entire graph $G$.
We quantify the connectivity of different patches of the graph $G$ using a separator \emph{profile} $s_G(r)$ defined as follows.
For $1 \leq r \leq n$, the separation profile corresponds to the separator of all subgraphs $\cH$ of $G$ of size at most $r$:
\begin{align}
    s_G(r) = \max_{\cH \subseteq G, \; |\cH| \leq r}\min_{\sep(\cH)}\left(|\sep(\cH)|\right)~.
\end{align}

In particular, it was shown that the graph separator constrains the distance $d$ and impacts the trade-off between $k$ and $d$.
Throughout this paper, we assume that the separator $s_G(r)$ is a polynomial, i.e.\ $s_G(r) \leq \beta r^c$ for some positive constant $\beta$ and constant $c \in (0,1]$.
Furthermore, we assume $\beta$ and $c$ are independent of $n$.
In light of this, we drop the subscript $n$ from the separation profile and simply write $s(r)$.

It was shown in \cite{baspin2022connectivity} that the results of Bravyi \& Terhal and by Bravyi, Poulin \& Terhal can both be generalized.
For a family $\{\cQ_n\}_n$ of $\dsl n,k,d \dsr$ quantum codes with connectivity graphs $G_n$ with max degree $\Delta$, these generalizations are both expressed in terms of the constant $c$:
\begin{align}
\label{eq:old-bound}
    d = O(\Delta n^{c})~, \quad kd^{2(1-c)} = O(n)~.
\end{align}
Thus, even without the geometry of Euclidean space, codes are constrained by some abstract notion of connectivity.
In \cite{baspin2022quantifying}, the implications of these bounds for constructing syndrome-extraction circuits in 2 dimensions was studied.
It was shown that constant-rate codes require a lot of long-range connectivity.

Two caveats are in order.
First, the bound on the distance requires that the connectivity graph $G_n$ has bounded degree.
Second, the trade-off between $k$ and $d$ does not match the BPT bound even in 2 dimensions.
Specifically, it is known that $s_G(r) = O(r^{1/2})$ for geometrically-local graphs in $\bbR^2$.
The $k$-$d$ trade-off in Eq.\ \eqref{eq:old-bound} implies that
\begin{align}
\label{eq:old-bound-2d}
    kd = O(n)~.
\end{align}
The mismatch between Eq.\ \eqref{eq:2d-bpt} and Eq.\ \eqref{eq:old-bound-2d} could mean one of two things---either Eq.\ \eqref{eq:old-bound-2d} is not tight or there exist codes with separators $s_G(r) = O(r^{1/2})$ that can achieve better parameters.
In this work, we address these problems.

\noindent First, we obtain a tighter trade-off between $k$ and $d$ that only depends on the size of the separator.
\begin{restatable}[Informal]{theorem}{kd-tradeoff}
\label{thm:main}
Let $\cQ$ be a $\dsl n,k,d \dsr$ quantum code whose connectivity graph $G$ has separation profile $s_G(r) \le O(r^c)$ for some $c \in (0,1]$. 
Then,
    \begin{align*}
        kd^{\frac{1-c^2}{c}} = O(n)~.
    \end{align*}
\end{restatable}

\noindent Second, we obtain a degree-oblivious bound on the distance in terms of the separation profile.
A similar result was obtained in \cite{baspin2022connectivity}.
That result only held for quantum LDPC codes, because it assumed that the connectivity graph had constant degree, but it used a milder assumption that the code had $s_G(r)\le O(r^c)$ only for $r=n$, rather than all $r$.
\begin{theorem}[Informal]
Let $\cQ$ be a $\dsl n,k,d \dsr$ quantum code whose connectivity graph $G$ has separation profile $s_G(r) \le O(r^c)$ for some $c \in (0,1]$. 
Then,
    \begin{align}
        d = O(n^c).
    \label{thm:main-2}
    \end{align}
\end{theorem}
\noindent This represents a strengthening of the Bravyi-Terhal bound as it is now independent of the degree of the connectivity graph $G$.

Both of these theorems follow from our main result, Lemma~\ref{lem:correct-1}.
Informally, this lemma states that under certain conditions, there are sets much larger than $d$ that can be erased without losing encoded information.
This was the intuition behind the BPT bound and Lemma~\ref{lem:correct-1} can be seen as a generalization of this idea dealing with codes that with connectivity graphs unconstrained by geometric locality.

However, there remains a mismatch between the BPT bound in Eq.\ \eqref{eq:2d-bpt} and Eq.\ \eqref{eq:new-bound-2d}.
It is still not clear if Theorem~\ref{thm:main} tight---for $2$-dimensional local codes, we find
\begin{align}
\label{eq:new-bound-2d}
    kd^{\frac{3}{2}} = O(n)~,
\end{align}
which can be obtained by substituting $c=1/2$ in Theorem~\ref{thm:main}.
To address this problem, we conclude this paper with a purely graph-theoretic conjecture.
If this conjecture holds, then it would imply that Theorem~\ref{thm:main} can be strengthened so as to strictly generalize the BPT bound from geometrically-local codes to a broader class of codes characterized by non-expanding connectivity graphs.
We state the conjectured generalization in Section~\ref{sec:conclusions}.

\textbf{Related work:} Flammia \emph{et al}.\ \cite{flammia2017limits} generalized the BPT bound to approximate quantum error correcting codes.
Even when allowing for imperfect recovery of encoded information, there is a trade-off between $k$ and $d$ imposed by geometric locality.
Delfosse extended the BPT bound to local codes in 2-dimensional hyperbolic space \cite{delfosse2013tradeoffs}.
Kalachev and Sadov recently studied and generalized the Cleaning Lemma, the central tool used in the Bravyi-Terhal and BPT bounds \cite{kalachev2022linear}.
In the context of circuits, Delfosse, Beverland \& Tremblay \cite{delfosse2021bounds} studied syndrome-extraction circuits subject to locality constraints.
They demonstrated that, for certain classes of quantum LDPC codes, there is a trade-off between the circuit depth and width (the total number of qubits used).
Most recently, Baspin, Fawzi \& Shayeghi \cite{baspin2023lower} demonstrated that there is a trade-off between the rate of sub-threshold error suppression, the circuit depth and the width.

\section{Background and notation}

\subsection{Stabilizer codes}
\label{ssec:codes}

In this paper we focus on stabilizer codes for simplicity.
Stabilizer codes are the equivalent of linear codes in the classical setting.
The ideas presented can be generalized to commuting projector codes (a superset of stabilizer codes) as in \cite{baspin2022connectivity}.

The state of pure states of a qubit is associated with $\bbC^{2}$ and the state of $n$ qubit pure states are associated with $(\bbC^{2})^{\otimes n}$.
Let $\cP = \{\ssI, \ssX, \ssY, \ssZ\}$ denote the Pauli group and $\cP_n = \cP^{\otimes n}$ denote the $n$-qubit Pauli group.
A stabilizer code $\cQ = \cQ(\cS)$ is specified by the stabilizer group $\cS$, an Abelian subgroup of the $n$-qubit Pauli group that does not contain $-\ssI$.
The code $\cQ$ is the set of states left invariant under the action of the stabilizer group, i.e.\ $ \cQ = \{ \ket{\psi} : \; \ssS \ket{\psi} = \ket{\psi} \;\forall \ssS \in \cS\}$.
Being an Abelian group, we can write $\cS$ in terms of $m = (n-k)$ independent generators $\{\ssS_1,...,\ssS_m\}$.

\begin{definition}[Connectivity graph]
	\label{def:connectivity}
	Let $\cQ = \cQ(\cS)$ be a code and let $\{\ssS_1,...,\ssS_m\}$ be a choice of generators for $\cS$.
	The connectivity graph $G = (V, E)$ associated with $\{\ssS_1,...,\ssS_m\}$ is defined as follows:
	\begin{enumerate}
		\item $V = [n]$, i.e. each vertex is associated with a qubit, and
		\item $(u,v) \in E$ if and only if there exists a generator $\ssS \in \{\ssS_1,...,\ssS_m\}$ such that $u,v \in \supp(\ssS)$.
	\end{enumerate}
  
  For a set $U \subseteq V$, we let $\overline{U}=V\setminus U$ denote the complement of $U$.
  We define the outer and inner boundaries of a subset $U \subseteq V$ as follows:
    \begin{enumerate}
        \item \textbf{Outer boundary:} $\bdry_{+} U$ is the set of all qubits $v \in \comp{U}$ such that for some $u \in U$, the edge $\{u,v\} \in E$.
        \item \textbf{Inner boundary:} $\bdry_{-} U$ is $\bdry_{+} \comp{U}$, the outer boundary of the complement of $U$.
        Equivalently, it is the set of all nodes $u \in U$ such that for some $v \not\in U$, the edge $\{u,v\} \in E$.
    \end{enumerate}
\end{definition}

In the language of graphs, we refer to subsets $U$ as a \emph{region} of the graph $G$.
The connectivity graph is not \emph{only} a function of the code $\cQ$, but also of the choice of generating set $\{\ssS_1,...,\ssS_m\}$.
Different generating sets can yield the same code but different graphs. 
Furthermore, in constructing the graph $G$, we discard information; if two qubits $u,v$ are in the support of a generator $\ssS$, we draw an edge between them regardless of whether the restriction of $\ssS$ to $u$ is $\ssX$, $\ssY$ or $\ssZ$.
Therefore, different codes might yield the same connectivity graph.
For the sake of brevity, we do not repeat that the graph is a function of the set of generators and assume that the set of generators is fixed and implicit.

In spite of this lack of uniqueness, the connectivity graph $G$ is a very useful tool.
For our purposes, it allows us to reason about \emph{correctability}, i.e.\ to what extent the code $\cQ$ is robust to erasure errors.
It can also be used for different purposes, see for example \cite{kovalev2013fault} or \cite{gottesman2014fault}.

\begin{definition}[Correctable set]
  For $U \subset V$, let $\cD[\overline{U}]$ and $\cD[V]$ denote the space of density operators associated with the sets of qubits $\overline{U}$ and $V$ respectively.
  The set $U$ is \emph{correctable} if there exists a recovery map $\cR: \cD[\comp{U}] \to \cD[V]$ such that for any code state $\rho_{\cQ} \in \cQ$, $\cR(\Tr_{U}(\rho_{\cQ})) = \rho_{\cQ}$.
\end{definition}

If $d$ is the distance of the code $\cQ$, then any set of size strictly less than $d$ is correctable.
However, some codes contain correctable regions that are much larger than $d-1$.
The connectivity graph makes it easy to represent and visualize this process of constructing large correctable regions.

\noindent In the following lemma, we present all the tools about correctable regions in connectivity graphs that we will use.
We refer to \cite{bravyi2009no,flammia2017limits} for proofs of these statements.

\begin{lemma}
  We restate some important results about correctable regions in an $\dsl n,k,d\dsr$ code.
\label{lem:correctable}
\begin{enumerate}
   \item \label{lem:trivial} \textbf{\emph{Trivial:}} Let $U \subset V$ be a correctable set.
   Any subset $W \subset U$ is correctable.
   \item \label{lem:dist} \textbf{\emph{Distance:}} If $U \subset V$ such that $|U| < d$, then $U$ is correctable.
   \item \label{lem:bptunion} \textbf{\emph{Union Lemma:}}
    Sets $U_1,\dots,U_\ell\subset V$ are \emph{decoupled} if the support of any generator $\ssS$ overlaps with at most one $U_i$. Equivalently, for all $i\neq j$, $U_i$ and $U_j$ are disjoint and disconnected in the connectivity graph, meaning there are no edges between $U_i$ and $U_j$.

    If $U_1,\dots,U_\ell$ are decoupled and each $U_i$ is correctable, then $\bigcup_{i=1}^\ell U_i$ is correctable.
    \item \label{lem:bptexpansion} \textbf{\emph{Expansion Lemma:}}
    Given correctable regions $U$, $T$ such that $ T \supset \bdry_+ U$, then $T \cup U$ is correctable. 
\end{enumerate}
\end{lemma}

\noindent As originally noted in \cite{bravyi2009no,bravyi2010tradeoffs}, the size and number of simultaneously correctable regions in the connectivity graph $G$ reveals a lot of information about the code $\cQ$.
We state the following lemma from \cite{bravyi2010tradeoffs}.
\begin{lemma}[Bravyi-Poulin-Terhal \cite{bravyi2010tradeoffs}, Eq.~14]
	\label{lem:bptabc}
	Consider an $\dsl n,k,d\dsr$ stabilizer code $\cQ$ defined on a set of qubits $Q$, $|Q| = n$, such that $Q = A \sqcup B \sqcup C$.
	If $A,B$ are correctable, then
	\begin{equation*}
		k  \leq |C|~.
	\end{equation*}
\end{lemma}

\subsection{Correctable sets in graphs}

With the tools above in Section~\ref{ssec:codes}, we now can reason about the connectivity graph directly.
In particular, all of the properties in Lemma~\ref{lem:correctable} can be understood entirely in terms of the connectivity graph.
As such, we make the following definition.
\begin{definition}[Correctable family]
  For a graph $G$, a \emph{$d$-correctable family} $\mathcal{F}$ is a family of subsets of vertices with the following properties.
  \label{def:correctable}
  \begin{enumerate}
    \item \textbf{{Trivial}:} If $U\in \mathcal{F}$, then any subset of $U$ is in $\mathcal{F}$.
    \item \textbf{{Distance}:} Every set $U$ with size $|U| < d$ is in $\mathcal{F}$.
    \item \textbf{{Union Lemma}:} If $U_1,\dots,U_\ell$ are disjoint and disconnected sets in $\mathcal{F}$, then $\cup_{i=1}^\ell U_i$ is in $\mathcal{F}$. 
    \item \textbf{{Expansion Lemma}:} If $U,T\in \mathcal{F}$ and $T\supset \bdry_+ U$, then $U\cup T\in \mathcal{F}$.
  \end{enumerate}
\end{definition}
\noindent With this definition, we can now define what it means for a set of vertices (qubits) to be correctable with respect to its graph, rather than with respect to the quantum code.
\begin{definition}[Correctable sets in graphs]
  A set $U$ is \emph{$d$-correctable with respect to graph $G$} if $U$ is in every $d$-correctable family of $G$.
  By abuse of notation, we may simply say $U$ is \emph{correctable} if $d$ and $G$ are understood.
\end{definition}
We observe that the family of $d$-correctable sets $U$ with respect to $G$ itself forms a $d$-correctable family, and it in fact forms the (unique) inclusion-minimum $d$-correctable family.
Equivalently, the family of $d$-correctable sets is the intersection of all $d$-correctable families.
Intuitively, a set $U$ is correctable in a graph if and only if it can be deduced to be correctable using only properties 1, 2, 3, and 4 above.

It follows immediately from Lemma~\ref{lem:correctable} that the family of correctable sets for a quantum code indeed form a $d$-correctable family with respect to the connectivity graph.
\begin{corollary}
  Let $\mathcal{Q}$ be $\dsl n,k,d\dsr$ quantum code with a connectivity graph $G=(V,E)$.
  Then the family of sets that are correctable (in the quantum code sense) form a $d$-correctable family for $G$.
  In particular, any subset $U\subset V$ that is correctable (with respect to the quantum code $\mathcal{Q}$) is $d$-correctable with respect to the graph $G$.
  \label{cor:correctable}
\end{corollary}

Corollary~\ref{cor:correctable}, together with Lemma~\ref{lem:bptabc}, gives us a framework for proving quantum code limitations using purely graph-theoretic arguments: show that for certain families of graphs, any $d$-correctable family contains two sets $A$ and $B$ such that $V\setminus (A\cup B)\le f(n,d)$, where $f(n,d)$ is our desired upper bound on the dimension of the code.

\subsection{Separators}
\label{subsec:sep}
In this section, we introduce the graph separator, a way to quantify the connectivity of the connectivity graph.

\begin{definition}
    Let $G = (V, E)$ be a graph.
	The separator of $G$, written $\sansserif{sep}(G)$ is the smallest set $T \subset V$ such there exists a disjoint partition of the vertices $V=U_1\sqcup T\sqcup U_2$ such that 
	\begin{enumerate}
    \item Both of $|U_1|, |U_2| \leq \frac{2}{3}|V|$, and
		\item There are no edges between $U_1$ and $U_2$.
	\end{enumerate}
\end{definition}

The separator is not uniquely defined as multiple sets could have the same size and still split the graph in two disjoint subgraphs.
However, this multiplicity does not affect our results, as it suffices to use the existence of \emph{one} such small set.
In this work, we prove limitations of quantum codes whose connectivity graphs are constrained by particular \emph{separation profiles}.

\begin{definition}
	\label{def:alphasep}[Separation Profile]
  For any graph $G$ on $n$ vertices, we define its separation profile $s_G : [1,...,n] \rightarrow [1,...,n]$, 
\[
    s_G(r) = \max\{ \sansserif{sep}(H) : \text{$H$ subgraph of $G$ with at most $r$ vertices}\}~.
	\]
\end{definition}
As an example, note that for a 2-dimensional grid graph $G$ of side length $L$, $s_G(r) = O(\sqrt{r})$ for all $1 \leq r \leq L^2$.
For a $D$-dimensional grid graph of side length $L$, $s_G(r) = O(r^c)$ for all $1 \leq r \leq L^D$ where $c = 1-1/D$.
Like these examples, we consider when the separation profile $G$ is polynomial, i.e.\ $s_G(r) = O(r^{c})$ for all $1 \leq r \leq n$.
This is a natural setting, as a variety of natural graph classes have polynomial separation profiles \cite{lipton1979separator, alon1990separator, gilbert1984separator, miller1991unified}.

\noindent The following is a useful property of graphs with polynomial separation profiles.
\begin{lemma}[Lemma 2.1 of \cite{henzinger1997faster}]
	\label{lem:rdivision}
  Let $\beta\ge 1$ and $c\in(0,1)$ be constants.
  Let $G$ have separation profile $s_G(r) \le \beta r^c$ for all $r$.
  There exists a constant $\alpha=\alpha(\beta,c)$ depending on $c$ such that for any $r$, there is a partition of the vertices of $G$ into sets $W_1,\dots,W_\ell$ such that
  \begin{itemize}
    \item Each $W_i$ has size at most $r$,
    \item $|\bdry_{-} W_i| \leq \alpha\cdot s_G(r)$
    \item $\ell\le \alpha \cdot n/r$.
  \end{itemize}
\end{lemma}
In words, this lemma states that the graph can be simultaneously partitioned into many sets $W_1,...,W_{\ell}$.
These sets can be large (of size $r$) and yet have small boundaries ($|\bdry_{-} W_i| \leq \alpha \cdot s_G(r)$).

\section{Proofs}

The main goal of this section is to prove Theorem~\ref{thm:main}.

\subsection{Warmup}

As a warmup, we show how to use our framework for quantum code limitations (Lemma~\ref{lem:bptabc} and Corollary~\ref{cor:correctable}) by proving a weaker version of Theorem~\ref{thm:main}, which appears as \cite[Theorem 23]{baspin2022connectivity}.
Our proof simplifies some aspects of their proof.

\begin{theorem}[Theorem 23 of \cite{baspin2022connectivity}]
  Let $\beta \ge 1$ and $c\in (0,1)$  be constants. 
  There exists a constant $\alpha=\alpha(\beta,c)>0$ such that the following holds for all positive integers $n,k,d$.
  Let $\cQ$ be an $\dsl n,k,d \dsr$ quantum codes with connectivity graph $G$
  If $G$ has separation profile $s_G(r) \leq\beta r^c$ for all $r$, then we have the following trade-off between $k$ and $d$:
  \label{thm:bk22}
  \begin{align}
      kd^{2(1-c)} \le \alpha n.
  \end{align}
\end{theorem}

\noindent By Lemma~\ref{lem:bptabc} and Corollary~\ref{cor:correctable}, Theorem~\ref{thm:bk22} follows from the following graph theoretic lemma.
\begin{lemma}[Implicit in \cite{baspin2022connectivity}]
  Let $\beta \ge 1$ and $c\in (0,1)$  be constants. 
  There exists a constant $\alpha=\alpha(\beta,c)>0$ such that the following holds for all positive integers $n\ge d$.
  Let $G$ be a graph on $n$ vertices with separation profile $s_G(r)\le \beta r^c$.
  Then there exists a partition $A\sqcup B\sqcup C$ of $G$ such that $A$ and $B$ are $d$-correctable with respect to $G$ and $|C|\le \alpha\cdot n/d^{2(1-c)}$.
  \label{lem:bk22}
\end{lemma}

To find the sets $A$ in $B$ above, we use the following lemma.
This lemma appears in \cite[Lemma 21]{baspin2022connectivity}, but we present the proof here for exposition, with a simpler analysis of the bound on the size of $V\setminus A$.

\begin{lemma}[\cite{baspin2022connectivity}]
  \label{lem:bk22-2}
  Let $\beta \ge 1$ and $c\in (0,1)$  be constants. 
  There exists a constant $\alpha=\alpha(\beta,c)>0$ such that the following holds for all positive integers $n\ge d$.
  Let $G$ be a graph on $n$ vertices with separation profile $s_G(r)\le \beta r^c$ for all $r$.
  There exists a vertex subset $A$ consisting of disjoint disconnected components of size at most $d-1$ such that $|V\setminus A|\le (\alpha n)/(d^{1-c})$.
\end{lemma}
    \begin{figure}
      \begin{center}
        \begin{minipage}[c]{0.55\textwidth}
        \begin{forest}
          [$V$ [$U_0$ [$U_{00}$ [$U_{000}$,fill=blue!20] [$U_{001}$,fill=blue!20]]  [$U_{01}$ [$U_{010}$,fill=blue!20] [$U_{011}$,fill=blue!20] ]][$U_1$ [$U_{10}$ [$U_{100}$,fill=blue!20] [$U_{101}$,fill=blue!20]][$U_{11}$,fill=blue!20 ]]]
        \end{forest}
        \end{minipage}
        \begin{minipage}[c]{0.38\textwidth}
        \begin{tikzpicture}[scale=0.3]
          \def\ua{4.8}
          \def\ub{5.3}
          \def\uxa{4.8}
          \def\uxb{5.2}
          \def\uya{6.2}
          \def\uyb{6.6}
          \fill[blue!20] (0,0) rectangle (10,10);
          \draw[fill=red!100] (0,\ua) rectangle (10,\ub);
          \node (sv) at (13,6) {$S(V)$};
          \draw[->] (sv) to[bend left=0] (10.3,{(\ua+\ub)/2});
          \draw[fill=red!50] (\uxa,0) rectangle (\uxb,\ua);
          \node (su0) at (5,-3) {$S(U_0)$};
          \draw[->] (su0) to [bend right=0] ({(\uxa+\uxb)/2},-0.3);

          \draw[fill=red!50] (\uya,\ub) rectangle (\uyb,10);
          \node (su1) at (9,13) {$S(U_1)$};
          \draw[->] (su1) to [bend right=10] ({(\uya+\uyb)/2},10.3);

          \draw[fill=red!20] (0,{0.4*\ua}) rectangle (\uxa,{0.45*\ua});
          \node (su01) at (14,{0.4*\ua}) {$S(U_{01})$};
          \draw[->] (su01) to[bend left=0] (10.3,{0.52*\ua});
          \draw[fill=red!20] (\uxb, {0.5*\ua}) rectangle (10,{0.55*\ua});
          \node (su00) at (-4,{0.4*\ua}) {$S(U_{00})$};
          \draw[->] (su00) to[bend right=0] (-0.3,{0.42*\ua});

          \draw[fill=red!20] (0,{0.5*\ua + \ub}) rectangle (\uya,{0.55*\ua+\ub});
          \node (su10) at (-4,{0.5*\ua+\ub}) {$S(U_{10})$};
          \draw[->] (su10) to[bend right=0] (-0.3,{0.52*\ua+\ub});

          \node at (8.6,7.5) {$U_{11}$};
          \node at (3.3,6.2) {$U_{101}$};
          \node at (3.3,8.8) {$U_{100}$};
          \node at (2.5,3.5) {$U_{001}$};
          \node at (2.5,1.0) {$U_{000}$};
          \node at (7.5,3.8) {$U_{011}$};
          \node at (7.5,1.3) {$U_{010}$};

        \end{tikzpicture}
        \end{minipage}
      \end{center}
      \caption{A binary tree obtained by recursively removing balanced separators. Each set is a subset of its parent. Each parent $U_*$ has two disjoint and disconnected children $U_{*0}$ and $U_{*1}$, which, together with the separator $S(U_*)$, partition $U_*$. The union of the \colorbox{blue!20}{blue sets} forms the set $A$ in Lemma~\ref{lem:bk22-2}.}
      \label{fig:tree}
    \end{figure}
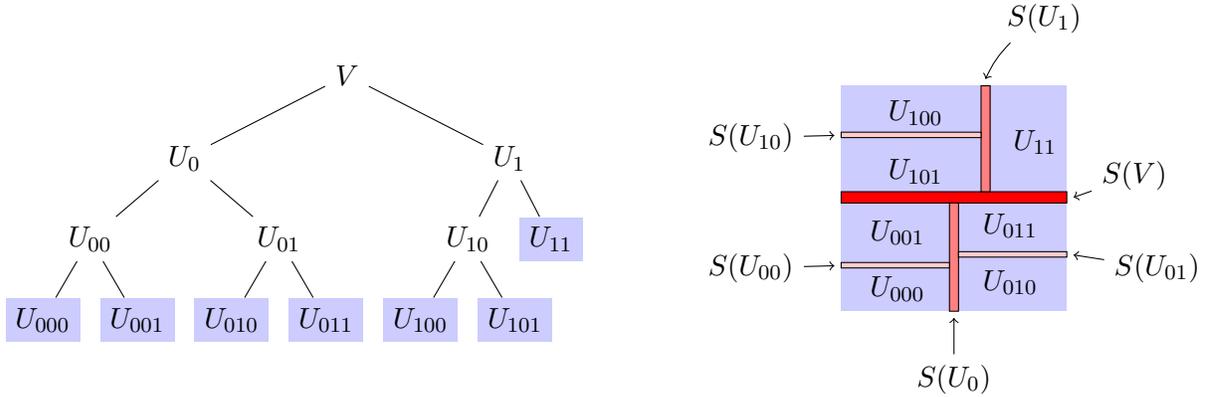
\begin{proof}
  This proof constructs a new graph $\cT$ from $G$.
  To avoid confusion, vertices of $\cT$ are referred to as nodes.
  
  Recursively construct a rooted binary tree $\mathcal{T}$ whose nodes are labeled by subsets of vertices of $G$.
  Start with the root node, labeled by the entire vertex set $V$.
  While there is a leaf node $U$ with size at least $d$, give node $U$ two children labeled $U_0$ and $U_1$, such that $U_0$ and $U_1$ are disjoint and disconnected vertex subsets, each of size at most $2|U|/3$, and such that $S(U)\defeq U\setminus (U_0\cup U_1)$ has size at most $\beta |U|^c$.
  Such children exists by the condition of the separation profile $s_G(|U|)\le \beta\cdot |U|^c$.
  Since each child is a strict subset of its parent, this process clearly terminates, giving a finite binary tree.
  See Figure~\ref{fig:tree} for an illustration. 

  We now choose vertex subset $A$ to be the union of all leaves of the tree, each of which has size at most $d-1$ by definition.
  Further, it is easy to check that all leaves are pairwise disjoint by construction.
  Thus, $A$ has the desired property. We now bound the size of $V\setminus A$.

  Since any child $U'$ of $U$ has size at most $|U'|\le \frac{2}{3}|U|$, and the only sets intersecting a node $U$ are its descendants and ancestors, we have that, for all $r$, the sets whose size is in $(2r/3,r]$ are all disjoint.
  This implies that there are at most $3n/2r$ sets with size in $(2r/3,r]$.
  
  \noindent Since all internal nodes in the tree have size at least $d$, we have
  \begin{align}
    |V\setminus A|
    &=  \sum_{U\text{ internal to }\mathcal{T}}^{} |S(U)| \nonumber\\
    &\le  \sum_{U\text{ internal to }\mathcal{T}}^{} \beta|U|^c \nonumber\\
    &\le  \sum_{\substack{r=(2/3)^in\\i=0,1,\dots,\floor{\log_{2/3}(d/n)}}}^{} \sum_{\substack{U\text{ internal to }\mathcal{T}\\ |U|\in (2r/3, r]}}^{} \beta r^c \nonumber\\
    &\le  \sum_{\substack{r=(2/3)^in\\i=0,1,\dots,\floor{\log_{2/3}(d/n)}}}^{} \frac{3n}{2r}\beta r^c
    \le \frac{\alpha \beta n}{d^{1-c}}
  \label{eq:ancestor}
  \end{align}
  where $\alpha=\alpha(c)$ is a constant that depends only on $c$,  and the last inequality holds because the sum is a geometric series, so the sum is within a constant factor of the largest term, which occurs when $r\sim d$.
\end{proof}

\noindent We now can prove Lemma~\ref{lem:bk22}, and thus Theorem~\ref{thm:bk22}.
\begin{proof}[Proof of Lemma~\ref{lem:bk22}]
  Apply Lemma~\ref{lem:bk22-2} to the graph $G$ and let $\alpha$ be the implied constant.
  We obtain a partition $A\sqcup \bar A$ of the vertex set $V$ such that $A$ has disconnected components of size at most $d-1$ and
  \begin{align}
  |\bar A|\le \alpha \cdot \frac{n}{d^{1-c}}.
  \end{align}
  Each component of $A$ correctable (Distance property of Definition~\ref{def:correctable}), so $A$ is correctable.
  Apply Lemma~\ref{lem:bk22-2} again to the subgraph $G[\bar A]$ of $G$ induced by $\bar A$ (the graph with vertex set $\bar A$ whose edges are the edges of $G$ with both endpoints in $\bar A$).
  Then there exists a partition $\bar A = B\sqcup C$ such that $B$ consists of disconnected components of size at most $d-1$ and
  \begin{align}
  |C|\le \alpha^2\cdot \frac{n}{d^{2(1-c)}}.
  \end{align}
  Since these components are disconnected in $G[\bar A]$, they are also disconnected in $G$.
  Each component is correctable, so $B$ is correctable.
  Thus, we have found our desired partition.
\end{proof}

\subsection{Proof of Theorem~\ref{thm:main} and Theorem~\ref{thm:main-2}}

We now prove Theorem~\ref{thm:main}. We establish Theorem~\ref{thm:main-2} along the way.
\begin{theorem*}[Theorem~\ref{thm:main}, formal]
  Let $\beta \ge 1$ and $c\in (0,1)$  be constants. 
  There exists a constant $\alpha=\alpha(\beta,c)>0$ such that the following holds for all positive integers $n,k,d$.
  Let $\cQ$ be a $\dsl n,k,d \dsr$ quantum code with connectivity graph $G$.
  If the separation profile of $G$ satisfies $s_G(r) \le \beta r^c$ for $c \in (0,1]$, then we have the following trade-off between $k$ and $d$:
    \begin{align*}
        kd^{\frac{1-c^2}{c}} \le \alpha n.
    \end{align*}
\end{theorem*}
To prove Theorem~\ref{thm:main}, we improve the graph-theoretic lemma, Lemma~\ref{lem:bk22}, to the following:
\begin{lemma}
  Let $\beta \ge 1$ and $c\in (0,1)$  be constants. 
  There exists a constant $\alpha=\alpha(\beta,c)>0$ such that the following holds for all positive integers $n\ge d$.
  Let $G$ be a graph on $n$ vertices with separation profile $s_G(r)\le \beta\cdot r^c$ for all $r$.
  Then there exists a partition $A\sqcup B\sqcup C$ of $G$ such that $A$ and $B$ are $d$-correctable and $|C|\le \alpha\cdot n/d^{(1-c)(1+1/c)}$.
  \label{lem:main}
\end{lemma}
\noindent By applying our framework (Lemma~\ref{lem:bptabc} and Corollary~\ref{cor:correctable}), this graph-theoretic lemma (Lemma~\ref{lem:main}) implies our desired quantum code bound (Theorem~\ref{thm:main}).

The key improvement is captured in Lemma~\ref{lem:correct-1}, which improves over Lemma~\ref{lem:bk22-2} when the sets have small boundary.
Lemma~\ref{lem:correct-1} states that sets $W$ of $d^{1/c}$ vertices with boundary $\ll d$ are correctable. 
To prove Lemma~\ref{lem:correct-1}, we recursively remove separators until the graph has connected components of size less than $d$. Each resulting component is correctable, and crucially the boundary of each component has size $<d$ and thus correctable, and you can recursively use the expansion lemma to deduce the whole set $W$ is correctable.
This is analogous to \cite{bravyi2010tradeoffs}, who showed that for geometrically local codes, $d\times d$ grids of qubits are correctable as their boundary has size $\sim d$.

    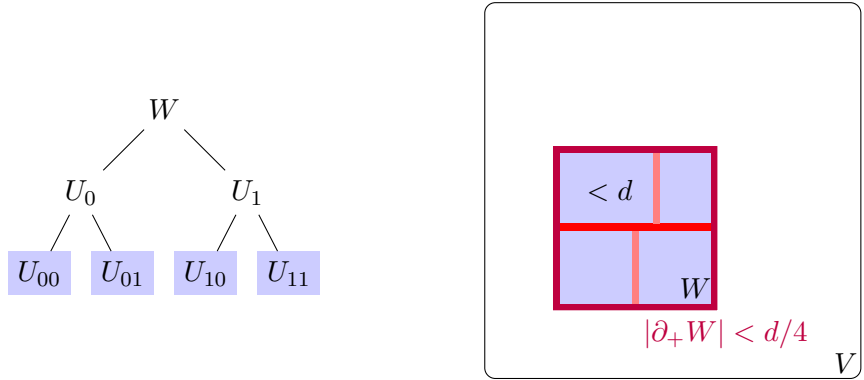
\begin{figure}
      \begin{center}
        \begin{minipage}[c]{0.5\textwidth}
          \begin{center}
            \begin{forest}
            [$W$ [$U_0$ [$U_{00}$,fill=blue!20] [$U_{01}$,fill=blue!20 ]][$U_1$ [$U_{10}$,fill=blue!20] [$U_{11}$,fill=blue!20 ]]]
            \end{forest}
          \end{center}
        \end{minipage}
        \begin{minipage}[c]{0.48\textwidth}
        \begin{tikzpicture}[scale=0.2]
          \def\ua{4.8}
          \def\ub{5.3}
          \def\uxa{4.8}
          \def\uxb{5.2}
          \def\uya{6.2}
          \def\uyb{6.6}
          \draw[rounded corners] (-5,-5) rectangle (20,20);
          \draw[color=white,fill=purple!100] (-0.5,-0.5) rectangle (10.5,10.5);
          \fill[blue!20] (0,0) rectangle (10,10);
          \fill[red!100] (0,\ua) rectangle (10,\ub);
          \fill[red!50] (\uxa,0) rectangle (\uxb,\ua);
          \fill[red!50] (\uya,\ub) rectangle (\uyb,10);

          \node at (19,-4) {$V$};
          \node at (9,1) {$W$};
          \node[text=purple!100] at (11,-2) {$|\partial_+ W| < d/4$};
          \node at (3.3,7.5) {$<d$};
        \end{tikzpicture}
        \end{minipage}
      \end{center}
      \caption{Lemma~\ref{lem:correct-1}. $\sim d^{1/c}$-sized vertex subsets $W$ with small boundary are correctable by the Expansion property.}
      \label{fig:tree-2}
    \end{figure}

\begin{lemma}
  Let $\beta \ge 1$ and $c\in (0,1)$  be constants.
  The following holds for all positive integers $n\ge d$.
  Let $G$ be a graph on $n$ vertices with separation profile $s_G(r)\le \beta r^c$ for all $r$.
  For $\varepsilon' = \frac{1-(2/3)^c}{20\beta}$, any set $W$ of at most $(\varepsilon d)^{1/c}$ vertices with outer boundary at most $d/4$ is $d$-correctable.
  \label{lem:correct-1}
\end{lemma}
\begin{proof}
  Recursively construct a rooted binary tree $\mathcal{T}$ whose nodes are labeled by subsets of vertices of $G$.
  Start with a root vertex, labeled by the set $W$.
  While there is a leaf $U$ with size at least $d$, give vertex $U$ two children labeled $U_0$ and $U_1$, such that $U_0$ and $U_1$ are disjoint and disconnected vertex subsets, each of size at most $2|U|/3$,  and such that $S(U)\defeq U\setminus (U_0\cup U_1)$ has size at most $\beta |U|^c$.
  Such children exists by the condition of the separation profile $s_G(|U|)\le \beta\cdot |U|^c$.
  Since each child is a strict subset of its parent, this process clearly terminates, giving a finite binary tree.
  See Figure~\ref{fig:tree-2}.

  By construction of the tree, the boundary of every $U$ in the tree $\mathcal{T}$ is a subset of $\partial_+ W \cup \bigcup_{\text{$U'$ ancestor of $U$}}^{} S(U')$.
  This is because $\partial_+ W$ contains all neighbors of vertices of $U$ outside $W$, and $\bigcup_{\text{$U'$ ancestor of $U$}}^{} S(U')$ contains all neighbors of $U$ inside $W$.
  Thus, $\partial_+ U$ has size at most
  \begin{align}
    |\partial_+ U|
    &\le |\partial_+ W| + \sum_{\text{$U'$ ancestor of $U$}}^{} |S(U')|  \nonumber\\
    &\le d/4 + \sum_{\text{$U'$ ancestor of $U$}}^{} \beta|U'|^{c}  \nonumber\\
    &\le d/4 + \sum_{i=0}^{\infty} \beta\left( \left(\frac{2}{3}\right)^i|W| \right)^c \nonumber\\
    &= d/4 + \frac{\beta|W|^c}{1-(2/3)^c}
    \le \left(1/4 + \frac{\varepsilon\beta}{1-(2/3)^c}\right) d < d/3.
    \label{eq:boundary}
  \end{align}
  The third inequality holds because the largest ancestor has size $|W|$, and each subsequent descendant's size is a factor less-than-$2/3$ smaller.

  We inductively show every set $U$ labeling a node in the tree $\mathcal{T}$ is correctable.
  The leaves of $\mathcal{T}$ are correctable because their size is $<d$.
  Suppose that a node $U$ in the tree is such that its children $U_0$ and $U_1$ are correctable (Distance property).
  Then $U_0\cup U_1$ is correctable as they are disjoint and disconnected (Union Lemma).
  Furthermore, by \eqref{eq:boundary}, $|\partial(U_0\cup U_1)| \le |\partial_+ U_0| + |\partial_+ U_1| \le d/3 + d/3 < d-1$ so $\partial(U_0\cup U_1)$ is correctable (Expansion Lemma).
  Thus, the set $\partial(U_0\cup U_1)\cup (U_0\cup U_1)$ is correctable, so $U$, which is a subset of $\partial(U_0\cup U_1)\cup (U_0\cup U_1)$, is also correctable (Trivial property).
  This completes the induction, showing that all nodes of the tree are correctable, so in particular $W$ is correctable.
\end{proof}

\noindent This result already allows us to prove Theorem~\ref{thm:main-2}, which we first restate here formally.
\begin{theorem*}[Theorem~\ref{thm:main-2}, formal]
Let $\beta \ge 1$ and $c\in (0,1)$  be constants. 
  There exists a constant $\alpha=\alpha(\beta,c)>0$ such that the following holds for all $k\ge 1$ and positive integers $n$ and $d$.
    Let $\cQ$ be a $\dsl n,k,d \dsr$ quantum code with connectivity graph $G$.
    If the separation profile of $G$ satisfies $s_G(r) \le \beta r^c$ for $c \in (0,1]$, then
    \begin{align}
        d \le \alpha n^c.
    \end{align}
\end{theorem*}
\begin{proof}[Proof of Theorem~\ref{thm:main-2}]
    Applying Lemma \ref{lem:correct-1} to the set of all vertices $V$, with the observation that $\bdry_{+} V = \emptyset$.
    If $d \geq (\varepsilon')^{-1} n^c$, then the set $V$ of all vertices is correctable.
    Then $k = 0$, which contradicts the $k>0$ assumption.
\end{proof}

Returning to Theorem~\ref{thm:main}, we now can use a known result on partitioning graphs into sets with small boundary, Lemma~\ref{lem:rdivision}.
It essentially states that we can obtain the guarantees of Lemma~\ref{lem:bk22-2} while additionally guaranteeing small boundary.
More precisely, we can remove separators until our graph consists of disjoint disconnected components $W_1,\dots,W_{\ell}$ of size $O(d^{1/c})$ and boundary $O(d)$, so that Lemma~\ref{lem:correct-1} applies and $W_1,\dots,W_{\ell}$ are each correctable.
Then by the Union lemma $W_1,\dots,W_{\ell}$ is correctable, and the number of vertices removed is, similar to Lemma~\ref{lem:bk22-2} with $d'=d^{1/c}$, $O(n/d^{(1-c)/c})$.

\begin{lemma}
  Let $\beta \ge 1$ and $c\in (0,1)$ be constants. 
  There exists a constant $\alpha=\alpha(\beta,c)>0$ such that the following holds for all positive integers $n\ge d$.
  Let $G$ be a graph on $n$ vertices where $s_G(r)\le \beta r^c$ for all $r$.
  There exists a $d$-correctable set $A$ such that $|V\setminus A|\le (\alpha n)/(d^{(1-c)/c})$.
  \label{lem:correct-2}
\end{lemma}
\begin{proof}
  Let $\alpha_{\ref{lem:rdivision}}=\alpha_{\ref{lem:rdivision}}(\beta,c)$ be given by Lemma~\ref{lem:rdivision}.
  Let $\varepsilon'=\varepsilon'(\beta,c)$ be given by Lemma~\ref{lem:correct-1}, and let $\varepsilon = \min(\varepsilon', 1/(4\beta \alpha_{\ref{lem:rdivision}}))$.
  By Lemma~\ref{lem:rdivision} with $r\defeq (\varepsilon d)^{1/c}$, there exists a partition $W_1',\dots,W_\ell'$ of the vertices of $G$ such that
  \begin{itemize}
    \item Each $W_i'$ has size at most $r\le (\varepsilon d)^{1/c}$,
    \item $|\partial_- W_i'|\le \alpha_{\ref{lem:rdivision}} \cdot s_G(r) \le d/4$, and
    \item $\ell\le \alpha_{\ref{lem:rdivision}} \cdot n/d^{1/c}$.
  \end{itemize}
  The bound on $|\bdry_{-} W_i'|$ follows from the assumption that $\varepsilon < 1/(4\beta\alpha_{\ref{lem:rdivision}})$ which implies
  \begin{align}
      \alpha_{\ref{lem:rdivision}} \cdot s_G(r) \leq \alpha_{\ref{lem:rdivision}} \cdot \beta (\varepsilon d) \leq d/4~.
  \end{align}
  For each $i=1,\dots,\ell$, let $W_i\defeq (W_i'\setminus \partial_- W_i')$.
  Note that the outer boundary $\partial_+ W_i$ is contained in $\partial_- W_i'$: if there were an edge $wv$ with $w\in W_i$ and $v$ in $\partial_+ W_i$ but outside of $\partial_- W_i'$, then $v$ is outside $W_i'$, which means $w$ is in $\partial_- W_i$.
  This contradicts the assumption of $w\in W_i$.

  Set $A = \bigcup_{i=1}^\ell W_i$. We show that $A$ is correctable.
  Since $\partial_+ W_i \subset \partial_- W_i'$, we have $|\partial_+ W_i|\le |\partial_- W_i'| \le d/4$.
  Further, $|W_i|\le |W_i'|\le (\varepsilon d)^{1/c}$.
  By assumption, $\varepsilon \leq \varepsilon'$ and therefore, we can apply Lemma~\ref{lem:correct-1} which implies that $W_i$ is correctable.
  By construction, $\partial_+ W_i$ is contained in $W_i'$, which is disjoint from $W_{j}'$, and thus $W_j$, for all $j\neq i$. 
  Hence, $W_i$ is disconnected from $W_j$ for all $i\neq j$. Therefore, by the Union Lemma, $A$ is correctable.

  \noindent To conclude, we bound the size of $V\setminus A$:
  \begin{align}
    |V\setminus A| = \abs{\bigcup_{i=1}^\ell \partial_- W_i'} 
    \le \frac{\alpha_{\ref{lem:rdivision}} \cdot n}{(\varepsilon d)^{1/c}}\cdot  \frac{d}{4}
    = \frac{\alpha n}{d^{(1-c)/c}}.
    \label{}
  \end{align}
  for a constant $\alpha=\alpha(\beta,c)$ depending only on $\beta$ and $c$.
\end{proof}

\noindent We now can prove Lemma~\ref{lem:main}, which implies Theorem~\ref{thm:main} by Lemma~\ref{lem:bptabc} and Corollary~\ref{cor:correctable}.
\begin{proof}[Proof of Lemma~\ref{lem:main}]
  Apply Lemma~\ref{lem:correct-2} to the graph $G$ and let $\alpha_{\ref{lem:correct-2}}$ be the implied constant depedning on $\beta$ and $c$.
  We obtain a correctable set $A$ such that 
  \begin{align}
  |V\setminus A|\le \alpha_{\ref{lem:correct-2}}\cdot  \frac{n}{d^{(1-c)/c}}. 
  \end{align}
  Apply Lemma~\ref{lem:bk22-2} to the subgraph $G[V\setminus A]$ of $G$ induced by $V\setminus A$ (the graph with vertex set $V\setminus A$ whose edges are the edges of $G$ with both endpoints in $V\setminus A$), and let $\alpha_{\ref{lem:bk22-2}}$ be the implied constant.
  Then there exists a partition $V\setminus A = B\sqcup C$ such that $B$ consists of disconnected components of size at most $d-1$ and
  \begin{align}
  |C|\le \alpha_{\ref{lem:correct-2}}\cdot \alpha_{\ref{lem:bk22-2}} \cdot \frac{n}{d^{(1-c)+(1-c)/c}}~.
  \end{align}
  Since these components are disconnected in $G[V\setminus A]$, they are also disconnected in $G$.
  Each component is correctable (Distance property) in Definition~\ref{def:correctable}, so $B$ is correctable.
  Thus, we have found our desired partition.
\end{proof}
\begin{remark}
  To obtain a stronger bound, one might try applying Lemma~\ref{lem:correct-2} twice, so that $|C|\le O(\frac{n}{d^{2(1-c)/c}})$.
  This in particular would match the BPT bound in \cite{bravyi2010tradeoffs}---for example, if $c=1/2$, this would give $|C|\le O(n/d^2)$, which yields the bound $kd^2\le O(n)$.
  However, a second application of Lemma~\ref{lem:correct-2} cannot be done, at least, in a black-box way.
  The issue is in the use of the Expansion Lemma.
  Given correctable sets $U$ and $T$, the Expansion Lemma implies that $U\cup T$ is correctable when $T$ contains the boundary of $U$ \emph{in the original connectivity graph}.
  The proposed second application of Lemma~\ref{lem:correct-2} would not work, because sets that are correctable with respect to the induced subgraph $G[V\setminus A]$ (particularly those that are correctable because of the Expansion Lemma) are not necessarily correctable with respect to the original graph $G$. 
  At the same time, we conjecture that this stronger bound $k \le O(n/d^{2(1-c)/c})$ can be proved with additional ideas. See Conjecture~\ref{conj:code} and Conjecture~\ref{conj:graph}.
\end{remark}

\section{Conclusions}
\label{sec:conclusions}

In this paper, we found a sharper trade-off between $k$ and $d$ for a quantum code $\cQ$ based on the separation profile of its connectivity graph $G$.
We also showed that the distance $d$ is upper bounded in terms of the separator of $G$ and independent of the degree of the graph $G$.
In this way, it improves on the Bravyi-Terhal bound from \cite{bravyi2009no}.

We observe that the bound in Theorem~\ref{thm:main} is still weaker than the Bravyi-Poulin-Terhal bound for $2$-dimensional local codes.
We conjecture that the Bravyi-Poulin-Terhal bound can be generalized completely from geometrically-local codes to a broader class of codes characterized by non-expanding connectivity graphs.

\begin{conjecture}[Graph theory conjecture]
  Let $\beta \ge 1$ and $c\in (0,1)$ be constants. 
  There exists a constant $\alpha=\alpha(\beta,c)$ such that the following holds for all $n\ge d$.
  Let $G$ be a graph with separation profile $s_G(r)\le \beta r^c$ for all $r$.
  Then there exists a partition $A\sqcup B\sqcup C$ of $G$ such that $A$ and $B$ are $d$-correctable with respect to graph $G$ and $|C|\le \alpha n/d^{2(1-c)/c}$.
  \label{conj:graph}
\end{conjecture}
 
Using our framework, we can prove the following theorem, which asserts that a stronger bound matching \cite{bravyi2010tradeoffs} follows from Conjecture~\ref{conj:code}.
This theorem is a consequence of Lemma~\ref{lem:bptabc} and Corollary~\ref{cor:correctable}.
\begin{theorem}[Quantum code bound conjecture]
  Assuming Conjecture~\ref{conj:graph}, the following holds.
  Let $\cQ$ be an $\dsl n,k,d \dsr$ quantum code.
  Let $G$ be the associated connectivity graph.
  For positive constants $\beta$ and $c \in (0,1)$, if this graph has a separator $s_G(r) \leq \beta r^c$, then we have the following trade-off between $k$ and $d$:
  \label{conj:code}
  \begin{align}
      kd^{2(1-c)/c} = O(n)~.
  \end{align}
\end{theorem}

Another direction is whether we can obtain better bounds using a more fine-grained description of the code.
Our bounds use a coarse description of the code, the connectivity graph, but there are more detailed and natural descriptions.
It is known that any quantum error correcting code can be described by a chain complex \cite{kitaev1997quantum}
\begin{align}
    C_2 \xrightarrow{\bdry_2} C_1 \xrightarrow{\bdry_1} C_0~,
\end{align}
where $C_i$ are $\bbF_2$ vector spaces and $\bdry_j$ are linear maps between spaces such that $\bdry_1 \circ \bdry_2 = 0$.
In this representation, qubits are associated with the elements of $C_1$ and $\ssX$ and $\ssZ$ stabilizer generators are associated with $C_2$ and $C_0$ respectively.
The relationship between the stabilizer generators and the qubits is specified by the boundary operators.
The adjacency matrix of the connectivity graph we use here is some function of \emph{both} $\bdry_2$ and $\bdry_1$.
It would be interesting to use graph invariants of the bipartite graphs specified by $\bdry_2$ and $\bdry_1$ to arrive at better bounds and tighter trade-offs.

\section{Acknowledgements}

NB is supported by the Australian Research Council via the Centre of Excellence in Engineered Quantum Systems (EQUS) project number CE170100009, and by the Sydney Quantum Academy.

\noindent VG is supported in part by a Simons Investigator award, a UC Berkeley Initiative for Computational Transformation award, and NSF grant CCF-2210823.

\noindent AK is supported by the Bloch Postdoctoral Fellowship from Stanford University and NSF grant CCF-1844628.

\noindent RL acknowledges support from the NSF Mathematical Sciences Postdoctoral Research Fellowships Program under Grant DMS-2203067, and a UC Berkeley Initiative for Computational Transformation award.

\bibliographystyle{abbrv}
\bibliography{references}

\end{document}